\documentclass{article}

\usepackage[top=2cm, bottom=2cm, left=2cm, right=2cm]{geometry}
\usepackage{graphicx}
\usepackage{algorithm}
\usepackage{algorithmic}
\usepackage{amsmath, amssymb}
\usepackage{amsthm}
\usepackage{amsfonts}
\usepackage{authblk}
\usepackage{cite}

\title{A Note on Enumeration by Fair Sampling}
\author[1,2,3,4]{Yuta Mizuno \thanks{mizuno@es.hokudai.ac.jp}}
\author[1,2,3]{Tamiki Komatsuzaki}
\affil[1]{Research Institute for Electronic Science, Hokkaido University}
\affil[2]{Graduate School of Chemical Sciences and Engineering, Hokkaido University}
\affil[3]{Institute for Chemical Reaction Design and Discovery (WPI-ICReDD), Hokkaido University}
\affil[4]{PRESTO, Japan Science and Technology Agency}
\newtheorem{problem}{Problem}
\newtheorem{lemma}{Lemma}
\newtheorem{theorem}{Theorem}

\begin{document}
    \maketitle

    \begin{abstract}
        This note describes an algorithm for enumerating
        all the elements in a finite set
        based on uniformly random sampling from the set.
        This algorithm can be used for
        enumeration by fair sampling with quantum annealing.
        Our algorithm is based on a lemma of the coupon collector's problem
        and is an improved version of the algorithm described in
        \textit{arXiv}:2007.08487 (2020).
        We provide a mathematical analysis and a numerical demonstration
        of our algorithm. 
    \end{abstract}

\section{Enumeration by fair sampling}
The objective of the algorithm described in this note is
to enumerate all the elements of a finite set
by repeating uniformly random sampling from the set.
This is mathematically formulated as follows.

\begin{problem} \label{problem}
    Suppose there exists a sampler that samples an element
    of a finite set $X$ with equal probability, that is,
    the probability that the sampler returns $x \in X$ is $1/|X|$, 
    where $|X|$ is the number of all the elements of $X$ and is unknown.
    Then, using the sampler repeatedly, enumerate all the elements of $X$
    with success probability greater than or equal to $1-\epsilon$,
    where $\epsilon \in (0,1]$ is a failure tolerance.
\end{problem}

This kind of problem appears in the application of quantum annealing
to enumeration problems \cite{Kumar2020}.
Quantum annealing is a procedure to find ground states of
an Ising Hamiltonian.
By designing a Hamiltonian so that
the ground state(s) of the spin configuration
corresponds to desired combinations,
such as cost minimum or constraint satisfactory combinations,
one can sample them by quantum annealing.
If the Hamiltonian system has multiple ground states
corresponding to elements of a finite set,
one can take advantage of the sampling ability of quantum annealing
to enumerate all the elements of the set.  

Although current versions of quantum annealing
do not always sample ground states with equal probability \cite{Konz2019}
and the research on
how to achieve the fair sampling in quantum annealing
is still one of the most interesting subjects on quantum annealing
\cite{Kumar2020, Konz2019, Yamamoto2020},
as the first step toward
the application of quantum annealing to enumeration problems,
we address a framework for enumeration by fair sampling
assuming that we have a fair sampler.
The same topic was addressed in the appendix of Ref. \cite{Kumar2020},
in which an algorithm for enumeration by fair sampling
was proposed based on a well-known lemma in the coupon collector's problem.
The coupon collector's problem is a classic problem in probability theory,
which is described as follows:
If there are $n$ different types of coupons and
we can get one of the types of coupons with equal probability at each trial,
how many trial do we need to collect all $n$ types of coupons?
The evaluation of the probability distribution of the number of trials
needed to collect all types of coupons
derives a termination condition of samplings in the enumeration.
In this note,
we derive an improved version of the algorithm in Ref. \cite{Kumar2020}
based on an extension of the lemma in the coupon collector's problem
used in Ref. \cite{Kumar2020}.

\section{Enumeration algorithm derived from the coupon collector’s problem}
\subsection{Algorithm description}
An algorithm for Problem \ref{problem} can be derived
based on the following lemma
(See Sec. \ref{subsec:proof-lemma1} for the proof).
\begin{lemma} \label{lemma:1}
    Let $T_m$ be the random variable representing
    the number of samplings necessary to collect $m$ different
    elements of $X$.
    If the true value of the number of all the elements of $X$ is $n$,
    for any positive integer $m\,(\le n)$ and
    any positive real number $\epsilon\,(\le 1/\mathrm{e})$,
    the tail distribution of $T_m$ is bounded from above as
    \begin{equation}
        \mathrm{P}\left(
            T_m > \left\lceil m\log\frac{m}{\epsilon} \right\rceil
            \,\middle|\, |X|=n
            \right)
            \le \epsilon.
    \end{equation}
\end{lemma}

This implies that, if the number of collected elements
after $\lceil m\log(m/\epsilon) \rceil$ samplings is less than $m$,
we can judge that no more different element exists in $X$
with failure probability less than or equal to $\epsilon$
\footnote{
    See the proof of Theorem \ref{theorem} below for the detailed discussion.
}.
In other words, we can use the condition
$T_m > \lceil m\log(m/\epsilon) \rceil$
as a termination condition of samplings.

\begin{algorithm}[t]
    \caption{Enumeration by Uniformly Random Sampling}
    \begin{algorithmic}[1] \label{algorithm}
        \renewcommand{\algorithmicrequire}{\textbf{Input:}}
        \renewcommand{\algorithmicensure}{\textbf{Output:}}
        \REQUIRE Sampling method $Sample()$,
                 failure tolerance $\epsilon \in (0,\frac{1}{\mathrm{e}}]$,\\
        \qquad checkpoints
               $\mathcal{C} = [m_1, \cdots, m_M] \in \mathbb{N}^M$
               s.t. $m_i<m_j$ for any $i<j$ and $m_M \ge |X|$
        \ENSURE Collection of samples $S$
        \STATE $S \leftarrow \{\}$
        \STATE $t \leftarrow 0$
        \FOR{$i=1, \cdots, M$}
            \WHILE{$t \le \lceil m_i\log(m_i M/\epsilon) \rceil$}
                \STATE $s \leftarrow Sample()$
                \STATE $S \leftarrow S\cup\{s\}$
                \STATE $t \leftarrow t+1$
            \ENDWHILE
            \IF{$|S|<m_i$}
                \STATE \textbf{break}
            \ENDIF
        \ENDFOR
        \RETURN $S$
    \end{algorithmic}
\end{algorithm}

Algorithm \ref{algorithm} is an algorithm based on Lemma \ref{lemma:1}.
In this algorithm, there are $M$ checkpoints
for checking the termination condition
$T_{m_i} > \lceil m_i\log(m_i M/\epsilon) \rceil$.
The maximum integer of the checkpoints, $m_M$,
should be larger than the unknown value of $|X|$
in order to ensure success in the enumeration.
Thus, one needs to estimate a (rough) upper bound of $|X|$
before applying the algorithm.
Note that the number of checkpoints, $M$, is contained in
the logarithm of the termination conditions
so that the failure probability of the judgement at each checkpoint
is less than or equal to $\epsilon/M$ rather than $\epsilon$.
This ensures that the total failure probability is
less than or equal to $\epsilon$.

The success probability and the number of times of sampling
for Algorithm \ref{algorithm} to collect all the elements of $X$ 
is formally stated as follows.
\begin{theorem} \label{theorem}
    Algorithm \ref{algorithm} outputs the collection of all the elements
    of $X$ with success probability greater than or equal to $1-\epsilon$,
    if the input satisfies the requirements
    $\epsilon \in (0,1/\mathrm{e}]$,
    $m_i<m_j$ for any $1 \le i < j \le M$,
    and $m_M \ge |X|$.
    In success cases,
    the number of times of sampling the algorithm calls is
    $\lceil m_k\log(m_k M/\epsilon) \rceil$, where
    index $k \ (1 \le k \le M)$ satisfies $m_{k-1} < |X| \le m_k$
    (let $m_{k-1}$ be $0$ for $k=1$).
\end{theorem}

\begin{proof}
    Let $k$ be the index of the checkpoint that satisfies
    $m_{k-1} < |X| \le m_k$, which always exists if $|X| \le m_M$.
    In the cases of success in the enumeration,
    the algorithm does not break out of the loop
    at all checkpoints before the $k$-th checkpoint,
    and collects all the elements of $X$
    before breaking out at the $k$-th checkpoint.
    Because the number of samplings
    the algorithm has done up to the $k$-th checkpoint is
    $\lceil m_k\log(m_k M/\epsilon) \rceil$,
    the second statement of the theorem is true.
    
    To prove the first statement,
    we evaluate the success probability under the condition $|X|=n$.
    From the above discussion, we can express the success probability as
    \[
        \mathrm{P}(\mathrm{success}) = 
        \mathrm{P}\left(
            \bigwedge_{i=1}^{k-1}
            \left\{
                T_{m_i} \le
                \left\lceil m_i\log\frac{m_iM}{\epsilon} \right\rceil
            \right\}
            \land
            \left\{
                T_n \le
                \left\lceil m_k\log\frac{m_kM}{\epsilon} \right\rceil
            \right\}
            \,\middle|\,
            |X|=n
        \right).
    \]
    Taking the negation, the failure probability can be expressed as
    \[
        \mathrm{P}(\mathrm{failure}) =
        \mathrm{P}\left(
            \bigvee_{i=1}^{k-1}
            \left\{
                T_{m_i} >
                \left\lceil m_i\log\frac{m_iM}{\epsilon} \right\rceil
            \right\}
            \lor
            \left\{
                T_n >
                \left\lceil m_k\log\frac{m_kM}{\epsilon} \right\rceil
            \right\}
            \,\middle|\,
            |X|=n
        \right).
    \]
    Due to the subadditivity of probabilities, we obtain
    \[
        \mathrm{P}(\mathrm{failure}) \le
        \sum_{i=1}^{k-1}
        \mathrm{P}\left(
            T_{m_i} >
            \left\lceil m_i\log\frac{m_iM}{\epsilon} \right\rceil
        \,\middle|\,
        |X|=n
        \right)
        +
        \mathrm{P}\left(
            T_n > \left\lceil m_k\log\frac{m_kM}{\epsilon} \right\rceil
            \,\middle|\,
            |X|=n
        \right). 
    \]
    Since $m_k \ge n$,  we get
    \begin{eqnarray*}
        &&
        \mathrm{P}\left(
            T_n > \left\lceil n\log\frac{nM}{\epsilon} \right\rceil
            \,\middle|\,
            |X|=n
        \right)\\
        &=&
        \mathrm{P}\left(
            \left\{
                T_n > \left\lceil m_k\log\frac{m_kM}{\epsilon} \right\rceil
            \right\}
            \lor
            \left\{
                \left\lceil m_k\log\frac{m_kM}{\epsilon}\right\rceil \ge T_n
                > \left\lceil n\log\frac{nM}{\epsilon} \right\rceil
            \right\}
            \,\middle|\,
            |X|=n
        \right) \\
        &=&
        \mathrm{P}\left(
            T_n > \left\lceil m_k\log\frac{m_kM}{\epsilon} \right\rceil
            \,\middle|\,
            |X|=n
        \right)
        +
        \mathrm{P}\left(
            \left\lceil m_k\log\frac{m_kM}{\epsilon}\right\rceil \ge T_n
            > \left\lceil n\log\frac{nM}{\epsilon} \right\rceil
            \,\middle|\,
            |X|=n
        \right) \\
        &\ge&
        \mathrm{P}\left(
            T_n > \left\lceil m_k\log\frac{m_kM}{\epsilon} \right\rceil
            \,\middle|\,
            |X|=n
        \right),
    \end{eqnarray*}
    which gives
    \[
        \mathrm{P}(\mathrm{failure}) \le
        \sum_{i=1}^{k-1}
        \mathrm{P}\left(
            T_{m_i} >
            \left\lceil m_i\log\frac{m_iM}{\epsilon} \right\rceil
        \,\middle|\,
        |X|=n
        \right)
        +
        \mathrm{P}\left(
            T_n > \left\lceil n\log\frac{nM}{\epsilon} \right\rceil
            \,\middle|\,
            |X|=n
        \right).
    \]
    According to Lemma \ref{lemma:1},
    all of $k$ terms in the right hand side of the above inequality are
    less than or equal to $\epsilon/M$.
    Therefore, the failure probability is
    less than or equal to $\epsilon$,
    under the condition $|X|=n$. 
    Because the above discussion is valid for any $n$,
    the first statement of the theorem is proven.
\end{proof}

The number of samplings required in our algorithm is
less than that in the algorithm described in Ref. \cite{Kumar2020}.
The algorithm described in Ref. \cite{Kumar2020} is based on
the special case of Lemma \ref{lemma:1} in which $m$ equals to $n$.
This special case of Lemma \ref{lemma:1} is
a well-known lemma in the coupon collector's problem.
Because the algorithm described in Ref. \cite{Kumar2020}
resets the sampling counter $t$ to 0
when $|S|$ becomes greater than $m_i$ for every $i \in \{1, \cdots, M\}$,
the algorithm requires additional samplings compared with our algorithm.
The expectation value of the number of the additional samplings 
equals to the expectation value of $T_{m_{k-1}+1}$ given $|X|=n$
\footnote{
    The expectation value of $T_{m_{k-1}+1}$ given $|X|=n$ is
    greater than $n\log[(n+1)/(n-m_{k-1}+2)]$.
    The proof is as follows.
    Let $t_i$ be the random variable representing
    the number of times of sampling
    necessary to obtain a new (uncollected) element
    after $i-1$ elements are collected.
    The probability distribution of $t_i$ is
    the geometric distribution with expectation $n/(n-i+1)$.
    Thus, for $m \le n$, we get
    \begin{eqnarray*}
        \mathrm{E}\left(T_{m}\,\middle|\,|X|=n\right)
        \ =\ 
        \mathrm{E}\left(\sum_{i=1}^{m}t_i\,\middle|\,|X|=n\right) 
        \ =\ 
        \sum_{i=1}^{m}\mathrm{E}\left(t_i\,\middle|\,|X|=n\right) 
        \ =\ 
        n\sum_{k=n-m+1}^n \frac{1}{k} 
        \ >\ 
        n\int_{n-m+1}^{n+1} \frac{\mathrm{d}x}{x} 
        \ =\ 
        n \log \frac{n+1}{n-m+1}.
    \end{eqnarray*}
    Substituting $m_{k-1}+1$ for $m$ completes the proof.
},
because the final reset of $t$ occurs
when $|S|$ becomes greater than $m_{k-1}$.
Our algorithm can save the additional samplings
thanks to the more general statement of Lemma \ref{lemma:1}.

\subsection{Proof of Lemma \ref{lemma:1}} \label{subsec:proof-lemma1}
In this subsection, we prove Lemma \ref{lemma:1}.
Before giving a proof of Lemma \ref{lemma:1},
we prove the special case for $m=n$:

\begin{lemma} \label{lemma:2}
    Suppose $X$ is a finite set with size $n$.
    Let $T$ be the random variable representing
    the number of times of samplings
    necessary to collect all $n$ different elements of $X$.
    Then, for any positive real number $\epsilon$,
    the tail distribution of $T$ is bounded from above as
    \begin{equation}
        \mathrm{P}\left(
            T > \left\lceil n\log\frac{n}{\epsilon} \right\rceil
            \,\middle|\,
            |X|=n
        \right)
        \le \epsilon.
    \end{equation}
\end{lemma}

\begin{proof}
    Let $S_\tau$ be the set of elements
    that have been already collected until time $\tau$.
    The probability that
    an element $x \in X$ has not been sampled yet up to the moment $\tau$ is
    \[
        \mathrm{P}\left(
            x \notin S_\tau
            \,\middle|\,
            |X|=n
        \right)
        =
        \left(1-\frac{1}{n}\right)^\tau
        \le
        \mathrm{e}^{-\frac{\tau}{n}}.
    \]
    Thus, the probability that $T > \tau$ can be evaluated as
    \begin{eqnarray*}
        \mathrm{P}\left(
            T > \tau
            \,\middle|\,
            |X|=n
        \right)
        &=&
        \mathrm{P}\left(
            \bigvee_{x \in X}
            \left\{
                x \notin S_\tau
            \right\}
            \,\middle|\,
            |X|=n
        \right) \\
        &\le&
        \sum_{x \in X}
        \mathrm{P}\left(
            x \notin S_\tau
            \,\middle|\,
            |X|=n
        \right) \\
        &\le&
        n\mathrm{e}^{-\frac{\tau}{n}}.
    \end{eqnarray*}
    Substituting $\tau = \lceil n\log(n/\epsilon) \rceil$, we obtain
    \[
        \mathrm{P}\left(
            T > \left\lceil n\log\frac{n}{\epsilon} \right\rceil
            \,\middle|\,
            |X|=n
        \right)
        \le \epsilon.
    \]
\end{proof}

Now, let us prove Lemma \ref{lemma:1}
as an extension of Lemma \ref{lemma:2}.

\begin{proof}
    Let $t_i$ be the random variable representing
    the number of times of sampling
    necessary to obtain a new (uncollected) element
    after $i-1$ elements are collected.
    This can be expressed as
    \[
        t_i = T_i - T_{i-1}.
    \]
    In the case that $t_i=l$ under the condition $|X|=n$,
    after $i-1$ elements are collected,
    the sampler returns
    any of the $i-1$ collected elements until the $(l-1)$-th trial and
    returns one of the $n-(i-1)$ uncollected elements at the $l$-th trial.
    Thus, the probability distribution of $t_i$ is
    the geometric distribution
    \[
        \mathrm{P}\left(
            t_i = l
            \,\middle|\,
            |X|=n
        \right)
        =
        \frac{n-(i-1)}{n}
        \left(\frac{i-1}{n}\right)^{l-1}.  
    \]
    Here, note that $t_1,\cdots,t_m$ are independent of each other.

    The random variable $T_m$ can be expressed as
    \[
        T_m = t_1 + t_2 + \cdots + t_m,
    \]
    and the tail distribution of $T_m$ can be written as
    \begin{eqnarray*}
        \mathrm{P}\left(
            T_m > \bar\tau
            \,\middle|\,
            |X|=n
        \right)
        &=&
        \sum_{\tau > \bar\tau}
        \sum_{\sum_{i=1}^m \tau_i = \tau}
        \mathrm{P}\left(
            t_1 = \tau_1, \cdots, t_m = \tau_m
            \,\middle|\,
            |X|=n
        \right) \\
        &=&
        \sum_{\tau > \bar\tau}
        \sum_{\sum_{i=1}^m \tau_i = \tau}
        n^{-\tau}
        \prod_{i=1}^m \left[n-(i-1)\right](i-1)^{\tau_i-1},
    \end{eqnarray*}
    where $\bar\tau$, $\tau$, and $\tau_i$ for $1 \le i \le m$ are positive integers
    and $\sum_{\sum_{i=1}^m \tau_i = \tau}$ means the summation with respect to all possible combinations of
    $\tau_1, \tau_2, \cdots, \tau_m$ such that $\sum_{i=1}^m \tau_i = \tau$
    \footnote{
        This condition implies the random variable $T_m$ equals to the integer $\tau$.
    }.

    We will prove the inequality
    \begin{equation} \label{eq:goal-ineq}
        \mathrm{P}\left(
            T_m > \bar\tau
            \,\middle|\,
            |X|=n
        \right)
        \le
        \mathrm{P}\left(
            T_m > \bar\tau
            \,\middle|\,
            |X|=m
        \right),
    \end{equation}
    under the condition
    \begin{equation}\label{eq:condition-inequality}
        m \le n,\quad 
        0 < \epsilon \le \frac{1}{\mathrm{e}},\quad
        \bar\tau \ge \left\lceil m \log \frac{m}{\epsilon} \right\rceil,
    \end{equation}
    so that the tail distribution of $T_m$ given $|X|=n$ is bounded from above
    by a probability that can be bounded using Lemma \ref{lemma:2}.
    To examine the dependence with respect to $n\ (\ge m)$ of the probability
    $\mathrm{P}\left(T_m > \bar\tau\,\middle|\,|X|=n\right)$,
    let us define functions $f:\mathbb{R}\to\mathbb{R}$
    and $g_\tau: \mathbb{R}\to\mathbb{R}$ as
    \begin{eqnarray*}
        f(x) &=& \prod_{i=1}^m \left[x-(i-1)\right], \\
        g_\tau(x) &=& x^{-\tau}f(x).
    \end{eqnarray*}
    If the function $g_\tau$ monotonically decreases with respect to $x\ (\ge m)$,
    Eq. (\ref{eq:goal-ineq}) is true as shown below.
    The first derivative of $g_\tau$ is
    \begin{eqnarray*} 
        g_\tau^\prime(x)
        &=&
        -\tau x^{-\tau-1} f(x)
        + x^{-\tau} \sum_{i=1}^m \frac{f(x)}{x-(i-1)} \\
        &=&
        x^{-\tau-1} f(x)
        \left[
            \sum_{i=1}^m \frac{x}{x-(i-1)} -\tau
        \right].
    \end{eqnarray*}
    Because $f(x) > 0$ for $x \ge m$,
    if the condition
    \[
        \tau \ge \sum_{i=1}^m \frac{x}{x-(i-1)}
    \]
    is satisfied for all $x \ge m$,
    the function $g_\tau$ decreases monotonically with respect to $x\ (\ge m)$.
    The right hand side of the above condition is
    bounded from above in the range $x \ge m$ as
    \begin{eqnarray*}
        \sum_{i=1}^m \frac{x}{x-(i-1)}
        &\le&
        \sum_{i=1}^m \frac{m}{m-(i-1)}\\
        &=&
        m \left( 1 + \sum_{k=2}^m \frac{1}{k} \right) \\
        &=&
        m \left( 1 + \int_1^m \frac{\mathrm{d}y}{\lceil y \rceil} \right)\\
        &\le&
        m \left( 1 + \int_1^m \frac{\mathrm{d}y}{y} \right) \\
        &=&
        m \log m + m \\
        &=&
        m \log \frac{m}{1/\mathrm{e}} \\
        &\le&
        m \log \frac{m}{\epsilon}\quad(\because \epsilon \le 1/\mathrm{e}).
    \end{eqnarray*}
    Thus, the condition
    \[
        m \le n,\quad
        0 < \epsilon \le \frac{1}{\mathrm{e}},\quad 
        \tau \ge \left\lceil m \log \frac{m}{\epsilon} \right\rceil
    \]
    is a sufficient condition for the monotonical decreasing of $g_\tau$.
    Under the condition (\ref{eq:condition-inequality}),
    the above sufficient condition for the monotonical decreasing of $g_\tau$
    is satisfied for any $\tau > \bar\tau$.
    Thus, we obtain the inequality
    \begin{eqnarray*}
        \mathrm{P}\left(
            T_m > \bar\tau
            \,\middle|\,
            |X|=n
        \right)
        &=&
        \sum_{\tau > \bar\tau}
        \sum_{\sum_{i=1}^m \tau_i = \tau}
        g_{\tau}(n) \prod_{i=1}^m (i-1)^{\tau_i-1} \\
        &\le&
        \sum_{\tau > \bar\tau}
        \sum_{\sum_{i=1}^m \tau_i = \tau}
        g_{\tau}(m) \prod_{i=1}^m (i-1)^{\tau_i-1} \\
        &=&
        \mathrm{P}\left(
            T_m > \bar\tau
            \,\middle|\,
            |X|=m
        \right),
    \end{eqnarray*}
    under the condition (\ref{eq:condition-inequality}).
    
    Substituting $\bar\tau = \lceil m\log(m/\epsilon) \rceil$,
    we get the inequality we want to prove:
    \[
        \mathrm{P}\left(
            T_m > \left\lceil m\log\frac{m}{\epsilon} \right\rceil
            \,\middle|\,
            |X|=n
        \right)
        \le
        \mathrm{P}\left(
            T_m > \left\lceil m\log\frac{m}{\epsilon} \right\rceil
            \,\middle|\,
            |X|=m
        \right)
        \le
        \epsilon,
    \]
    where the second inequality is due to Lemma \ref{lemma:2}.
\end{proof}

\section{Numerical tests and discussions}

\begin{figure}
    \begin{center}
        \includegraphics[width=17cm]{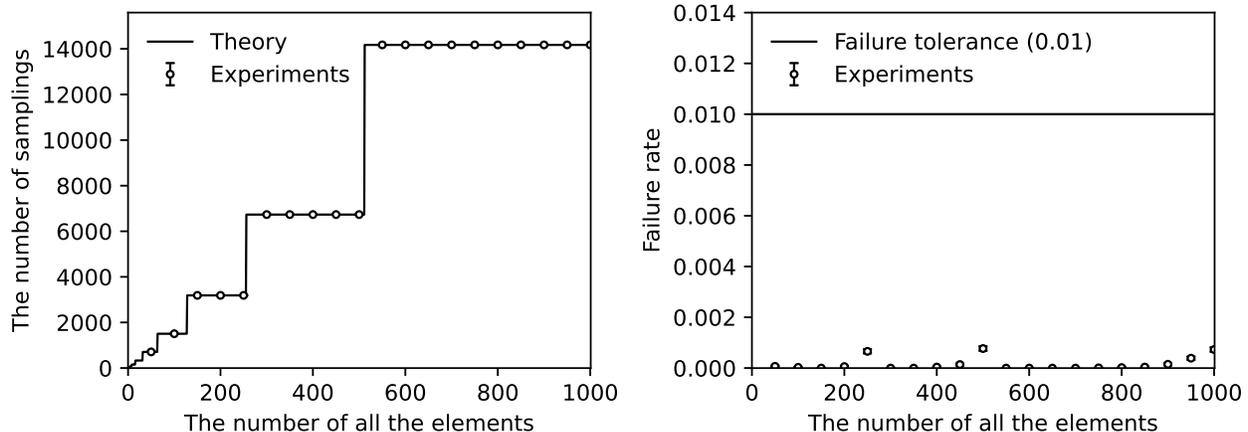}
    \end{center}
    \caption{
        Results of numerical tests for Algorithm \ref{algorithm}.
        The left panel shows the numbers of
        times of sampling in success cases.
        The circles with bars denote
        the sample means and standard deviations of the numerical results
        and the solid line denotes the theoretical result
        in Theorem \ref{theorem}.
        Note that, because the number of samplings in success cases exactly equals to
        $\lceil m_k\log(m_kM/\epsilon) \rceil$ in theory (Theorem \ref{theorem}),
        the standard deviations of the numerical results are 0 and
        the bars denoting the standard deviations are invisible in the left panel.
        The right panel shows the failure rates of the algorithm.
        The circles with bars denote
        the failure rates and their standard errors of the numerical results
        and the solid line denotes the failure tolerance $\epsilon$
        which is an input of the algorithm.
        Note that, because the number of test runs is large enough,
        the bars denoting the standard error are
        short and almost invisible in the right panel.
        In the numerical tests,
        the checkpoints were $\mathcal{C}=[2^1,2^2,\cdots,2^{10}]$,
        the failure tolerance $\epsilon$ was 0.01,
        the numbers of all the elements $|X|$ for test cases were
        $50, 100, 150, \cdots, 1000$, and
        the number of times of test runs for each case was $10^5$.
    }\label{fig:test-result}
\end{figure}

Results of numerical tests are shown in Fig. \ref{fig:test-result}.
In the numerical tests,
the checkpoints were $\mathcal{C}=[2^1,2^2,\cdots,2^{10}]$
(the same as those in the algorithm in Ref. \cite{Kumar2020}),
the failure tolerance $\epsilon$ was 0.01,
the numbers of all the elements $|X|$ for test cases were
$50, 100, 150, \cdots, 1000$, and
the number of times of test runs for each case was $10^5$.

The left panel shows the numbers of times of sampling
in success cases.
The numerical results (circle)
and the theoretical result (solid line) are
in good agreement with each other as expected.

The right panel shows the failure rates of the algorithm.
For all test cases,
the failure rates were less than the failure tolerance $\epsilon$,
which demonstrates the validity of Algorithm \ref{algorithm}.
Moreover,
the failure rates were much smaller than the failure tolerance and
for most cases they were almost zero.
This is because
the inequality in Lemma \ref{lemma:1} is not tight.
In the proof of Lemma \ref{lemma:1},
we proved the inequality
\[
    g_\tau(n) \le g_\tau(m).
\]
The tightness of the inequality is represented
by the ratio of $g_\tau(n)$ and $g_\tau(m)$,
\[
    \rho_\tau \equiv \frac{g_\tau(n)}{g_\tau(m)}
              = \binom{n}{m} \left(\frac{m}{n}\right)^\tau,
\]
which can be extremely smaller than 1
as shown in Fig. \ref{fig:tightness}.
This figure additionally shows that
the inequality is tight when $m \simeq n$.
This may be the reason why
the failure rates in some test cases where $m_k \simeq n$
are relatively large.

The looseness of the inequality implies that
we can derive more efficient algorithms for Problem \ref{problem}
based on a tighter inequality evaluation.
Further improvements of the efficiency (the number of times of sampling)
and the adaptation for non-uniform samplings of
the sampling-based enumeration algorithm
remains as future works.

\begin{figure}
    \begin{center}
        \includegraphics[width=8cm]{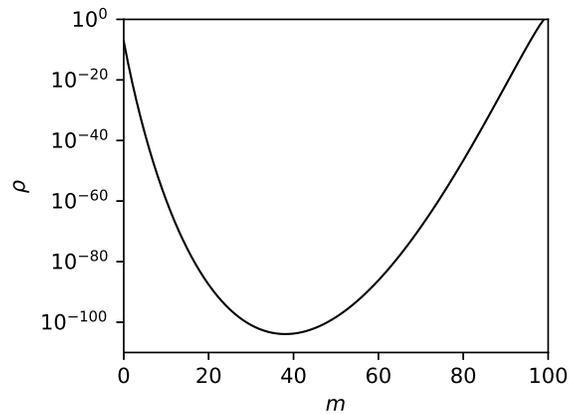}
    \end{center}
    \caption{
        The tightness of the inequality
        which is the basis of Algorithm \ref{algorithm}.
        The tightness is represented by $\rho_\tau$,
        that is, the ratio of $g_\tau(n)$ and $g_\tau(m)$.
        In this figure, $n=100$ and
        $\tau=\lceil m\log(m/\epsilon) \rceil$ with $\epsilon=0.01$.
        Note that the vertical axis is log-scaled.
    }\label{fig:tightness}
\end{figure}

\section*{Acknowledgements}
This work is supported by JST, PRESTO Grant Number JPMJPR2018.


\begin{thebibliography}{9}
    \bibitem{Kumar2020}
    V. Kumar, C. Tomlin, C. Nehrkorn, D. O'Malley, and J. Dulny III,
    ``Achieving fair sampling in quantum annealing'',
    \textit{arXiv}:2007.08487 (2020).
    \bibitem{Konz2019}
    M. S. K\"{o}nz, G. Mazzola, A. J. Ochoa, H. G. Katzgraber, and M. Troyer,
    ``Uncertain fate of fair samplng in quantum annealing'',
    \textit{Phys. Rev. A} \textbf{100}, 030303 (2019).
    \bibitem{Yamamoto2020}
    M. Yamamoto, M. Ohzeki, and K. Tanaka,
    ``Fair Sampling by Simulated Annealing on Quantum Annealer'',
    \textit{J. Phys. Soc. Jpn.} \textbf{89}, 025002 (2020).
\end{thebibliography}
\end{document}